\documentclass[lettersize,onecolumn, journal]{IEEEtran}
\usepackage{xcolor}
\usepackage{verbatim}

\usepackage{amsthm}

\usepackage[overload]{empheq} 

\newtheorem{definition}{Definition}

\newtheorem{lemma}{Lemma}

\newtheorem{theorem}{Theorem}

\usepackage{units}

\usepackage{amsmath,amsfonts}
\usepackage{amssymb}

\usepackage{algorithmic}
\usepackage{algorithm}
\usepackage{array}
\usepackage[caption=false,font=normalsize,labelfont=sf,textfont=sf]{subfig}
\usepackage{textcomp}
\usepackage{stfloats}
\usepackage{url}
\usepackage{verbatim}
\usepackage{graphicx}
\usepackage{cite}
\begin{document}


\title{
Computation-Limited Signals: A Channel Capacity Regime Constrained by Computational Complexity
}

\author{Saulo Queiroz, 
João P. Vilela, and Edmundo Monteiro, \IEEEmembership{Senior IEEE}
\thanks{This work is funded by The Science and Technology Development Fund, Macau SAR. (File no. 0044/2022/A1) 
and Agenda Mobilizadora Sines Nexus (ref. No. 7113), supported by the Recovery and Resilience Plan (PRR) and by the European Funds Next Generation EU.}
\thanks{Saulo Queiroz (sauloqueiroz@utfpr.edu.br) is with the Academic Department of Informatics of the Federal University of Technology Paran\'a (UTFPR), Ponta Grossa, PR, Brazil, the Centre for Informatics and Systems of the University of Coimbra (CISUC) and the University of Porto, Portugal. 
}
\thanks{Jo\~ao P. Vilela (jvilela@fc.up.pt) is with CRACS/INESCTEC, CISUC and the Department of Computer Science, Faculty of Sciences, University of Porto, Portugal.}
\thanks{Edmundo Monteiro (edmundo@dei.uc.pt) is with the Department of Informatics Engineering of the University of Coimbra and CISUC, Portugal.}
}


\markboth{Journal of \LaTeX\ Class Files,~Vol.~XX, No.~X, XXXX~XXXX}%
{Queiroz \MakeLowercase{\textit{et al.}}: Computation-Limited Signals: A Channel Capacity Regime Constrained by Computational Complexity}


\maketitle
\begin{abstract}
In this letter, we introduce the computation-limited (comp-limited) signals, 
a communication capacity regime where the computational complexity of signal processing 
is the primary constraint for communication performance, overriding factors such as power or bandwidth.
We present the Spectro-Computational (SC) analysis, a novel mathematical framework designed to enhance classic concepts of information theory --such as data rate, spectral efficiency, and capacity -- to accommodate the computational complexity overhead of signal processing.
We explore a specific Shannon regime where capacity is expected to increase indefinitely with channel resources.  
However, we identify conditions under which the time complexity overhead can cause capacity to decrease rather than increase, leading to the definition of the comp-limited signal regime. 
Furthermore, we provide examples of SC analysis and demonstrate that the OFDM waveform falls under the comp-limited regime unless the lower-bound computational complexity of the $N$-point DFT problem verifies as $\Omega(N)$, which remains an open challenge in the theory of computation.
\end{abstract}


\section{Introduction}
\IEEEPARstart{I}{nformation} theory introduces power and bandwidth 
as the fundamental resources to describe the capacity of a noisy
channel.
The development of clever physical layers, coupled with the adoption of larger spectrum resources, has enabled unprecedented data rates. Consequently, the computational resources required to process more bits per signal have grown accordingly, highlighting the trade-off between signal processing complexity and data rate.
Despite that, as far as we know, little knowledge have been
produced to correlate these performance indicators.

Some research
efforts propose unified models of computation and information theory but
without concerning about the interplay between the signal processing time complexity and capacity. 
For instance, works such as \cite{turingshannon-2019} concern about 
whether a discrete (Turing) machine is \emph{able} to compute a given channel 
capacity function. Other works bring the term ``complexity'' to information theory  
but with different meaning than that of the computation complexity.
This is the case of the communication complexity theory~\cite{raoyehudayoff2020}
and the ``Kolmogorov complexity''~\cite{chaitin_1987},
in which the term stands for the minimum number of message exchanges to solve a problem
in a distributed manner and the length of the irreducible form of an information, 
respectively.

{\color{black}
In this letter, we present an analytical framework referred to as the ``Spectro-Computational (SC) analysis''. 
With the SC analysis, we revisit classic concepts of information theory -- such as throughput, spectral 
efficiency and capacity -- in order to account for the signal processing computational complexity overhead. 
Our mathematical framework enhances and generalizes concepts we previously 
propose to study the complexity-throughput trade-off lying in the context of 
specific waveforms~\cite{fastenough-2022},~\cite{queiroz-wcl-19},~\cite{queiroz-access-2020},~\cite{queiroz-cost-ixs-19}.
Based on that, we refer to a specific Shannon capacity regime to derive a novel capacity 
regime in which computational complexity matters more than channel resources such as
bandwidth or received power.
}

The remainder of this letter is organized as follows. 
In Section~\ref{sec:rationale}, we review the background
and present the rationale for our mathematical framework.
In Section~\ref{subsec:scanalysis}, we present the mathematical framework of 
the SC analysis. In Section~\ref{sec:complimited}, we formalize the
comp-limited communication regime. In Section~\ref{sec:examples}, we
present practical examples of the SC analysis. In 
Section~\ref{sec:conclusion} we present our conclusion and future work.

{\color{black}\section{Background and Rationale}\label{sec:rationale} 
In this section, we review some key properties of asymptotic notation
that will support our analyses throughout this work (subsection~\ref{subsec:asymptotic}).
In subsection \ref{subsec:rationale}, we present the rationale of our proposal 
by discussing the interplay between computational complexity and channel resources in 
the Shannon communication system.

\subsection{Asymptotic Notation}\label{subsec:asymptotic}
The asymptotic analysis relate functions $f(N)$ and $g(N)$ as $N\to\infty$.
For the quantities of this work, we assume increasing non-negative functions. 
We follow the classic notation popular
in the literature of analysis of algorithms.
Thus, if $f(N)=\Theta(g(N))$, $f(N)=o(g(N))$ and $f(N)=\omega(g(N))$, it denotes 
that the order of growth of $f(N)$ is equal to, strictly lower than, or strictly 
higher than the order of growth of $g(N)$, respectively. 
{\color{black} Based on these notations, one can also define $O(.)$ and $\Omega(.)$ as 
follows,
\begin{eqnarray}
f(N)=O(g(N)) \Rightarrow [f(N)&=&o(g(N)) \text{ or } \nonumber\\
f(N)&=&\Theta(g(N)) ] \label{eq:o}\\
f(N)=\Omega(g(N)) \Rightarrow [f(N)&=&\omega(g(N)) \text{ or } \nonumber\\
f(N)&=&\Theta(g(N))]  \label{eq:omega}
\end{eqnarray}
}
Thus, these asymptotic notations can be defined according to Eqs.~(\ref{eq:asymptoticnotation}),
assuming existing limits and a real constant $c>0$.
\begin{subequations}
    \begin{empheq}[left={\displaystyle\lim_{N\to\infty}\frac{f(N)}{g(N)}=\empheqlbrace\,}]{align}
         c, & \quad \text{if $f(N)=\Theta(g(N))$} \label{eq:prop1}\\    
         0,  & \quad \text{if $f(N)=o(g(N))$} \label{eq:prop2}\\
         \infty,   &  \quad \text{if $f(N)=\omega(g(N))$ \label{eq:prop3}} 
    \end{empheq} \label{eq:asymptoticnotation}
\end{subequations}

\subsection{A Case for a Time Complexity-Constrained Signal Regime}\label{subsec:rationale}
In this subsection, we firstly review the channel resources considered by Shannon to describe
the capacity regimes. Then, we argue how these channel resources are related to the signal
processing computational complexity and argue for a time complexity-constrained capacity regime.

\subsubsection{Shannon Capacity Regimes}\label{subsec:shannonregimes}
Let us consider the Additive Gaussian White Noise (AWGN) channel capacity formula 
of Shannon, based on which the two classic channel capacity regimes of information 
theory derives from, namely, the Bandwidth-Limited Regime (BLR) and Power-Limited Regime (PLR). 
According to Shannon, the capacity of an AWGN channel is
\begin{eqnarray}
 C &=& W\log_2(1+\mathsf{SNR}) \quad \unit{bits/second} \label{eqn:shannon}\\
\mathsf{SNR} &=& \frac{\mathcal{P}}{WN_0} \label{eqn:snr}
\end{eqnarray}
where $W$ is the channel bandwidth,
$\mathsf{SNR}$ is the Signal-to-Noise Ratio (SNR),
$\mathcal{P}$ is the received signal power and $N_0$ is the
noise power spectral density. Eq.~(\ref{eqn:shannon}), establishes an upper 
bound for the data rate $R$ experienced by a $B(W)$-bit message in an 
AWGN channel. In other words, for a symbol period of 
$T_{\text{sym}}$, data rate $R$ and the spectral efficiency $SE$ are
given in Eq.~(\ref{eq:datarate}) and Eq.~(\ref{eq:se}), respectively.
\begin{eqnarray}
  R &=& \frac{B(W)}{T_{\text{sym}}} < C \quad \text{bits/second}\label{eq:datarate}\\
  SE &=& \frac{R}{W} \quad \text{bits/second/Hertz}\label{eq:se}
\end{eqnarray}
in which $T_{\text{sym}}$ is the signal period.
PLR results when SNR is very small ($\mathsf{SNR}\ll 1$).
In this case, one can approximate $\log_2(1+\mathsf{SNR})$ as $\mathsf{SNR}\log_2 e$, 
thereby $C$ becomes linear on $\mathcal{P}$ and is not affected by $W$. 
If SNR remains high ($\mathsf{SNR}\gg 1$) as $W$ grows, $C$ becomes proportional to $W$, 
leading to the BLR case. Note, however, that widening $W$ for a fixed $\mathcal{P}$
impairs SNR due to the resulting overall noise. In this case, the regime changes from
BLR to PLR as $W$ grows. 

\subsubsection{Channel Resources and Time Complexity}
Beyond power and spectrum, computational resources are also intrinsic to 
Shannon's communication system. In this system, a transmitter/receiver must be able to
``\emph{operate} on the message in some way to produce a signal''~\cite{Shannon48}. 
{\color{black} Specifically, in digital communication systems (the focus of this work), the computational resources required for this processing depend on the length of the transmitting message, which, in turn, is determined by the available channel resources.} Therefore, capacity regimes such as BLR and PLR directly affect 
the implied computational complexity. By its turn, such complexity overhead 
can impair time-related performance indicators (e.g., throughput, capacity) if
it is not neglected in the analysis.

\subsubsection{What If Channel Resources Grow Arbitrarily?}\label{subsubsec:snrregime}
To illustrate how time complexity can affect communication performance,
 consider the Shannon capacity formula (Eq.~\ref{eqn:shannon}) under the fixed 
SNR regime, {\color{black}i.e., $\mathsf{SNR}=c$, for a real constant $c>0$} (Eq.~\ref{eqn:snr}).
{\color{black}In this regime}, as $W\to\infty$,
$\mathcal{P}\to\infty$ accordingly to counter the resulting noise and keep the SNR constant.
Thus, in this case, $C=cW$ for some constant $c>0$, i.e., $C=\Theta(W)$ (Eq.~\ref{eq:prop1}).
{\color{black} In practice, this means that one might expect increasing system performance 
if more resources are assigned to the channel. Formally,
}
$\lim_{W\to\infty, \mathcal{P}\to\infty}C = \infty \quad \text{bits/second}$.

Let us now {\color{black}consider} the effect of the time complexity in the analysis.
Let $T(W)$ denote the number of computational instructions required to turn the $B(W)$-bit 
message into a $W$ Hertz signal. For a finite baseband processor (we further discuss this 
assumption in Section~\ref{subsec:scanalysis}), each instruction takes a runtime of
$t>0$ seconds. 
If one accounts for the time complexity overhead in the throughput,
the classic data rate formula of Eq.~(\ref{eq:datarate}) rewrites as
\begin{eqnarray}
  R_{\text{comp}}(W) &=& \frac{B(W)}{tT(W)+T_{\text{sym}}} < C \quad \text{bits/second}\label{eq:dataratecomp}
\end{eqnarray}

Under arbitrarily large channel resources, 
{\color{black}both $B(W)$ and $T(W)$ tend} to infinity. 
In this case, since $t$ does not depend on $W$, and $T_{\text{sym}}$ is 
bounded by a constant -- to ensure higher sampling rate as $W$ grows --,
the computational-constrained data rate solely depends on the ratio $B(W)/T(W)$ 
as $W\to\infty$. If $B(W)$ grows faster than $T(W)$, i.e., $B(W)=\omega(T(W))$,
then, by definition of Eq.~(\ref{eq:prop3}), $R_{\text{comp}}({W\to\infty})=B(W)/T(W)=\infty$.
{\color{black} This means that complexity does not become a bottleneck for the
overall system performance since $R_{\text{comp}}(W)$ grows arbitrarily on 
$W$, just as the non complexity-constrained data rate $R$ (Eq.~\ref{eq:datarate}) does. 

Conversely}, 
if $B(W)=o(T(W))$ then $R_{\text{comp}}(W\to\infty)=B(W)/T(W)=0$ (Eq.~\ref{eq:prop2}). 
{\color{black} 
In other words, in the fixed SNR regime of information theory, 
the upper bound for $R_{\text{comp}}(W)$ grows linearly on $W$. By contrast, when the computational 
complexity is considered in the $B(W)=o(T(W))$ case, that Shannon capacity bound becomes 
meaningless, since $T(W)$ causes $R_{\text{comp}}(W)$ to behave as a decreasing 
function of $W$. Moreover, if $T(W)$ is asymptotically optimal, this effect
 cannot be reversed except by increasing the computational resources as $T(W)$ grows.
Therefore, there might exist regimes in which the maximum achievable data rate
is limited by computational -- rather than channel -- resources.
}

\begin{figure*}[t]
\centering
  \includegraphics[scale=0.3]{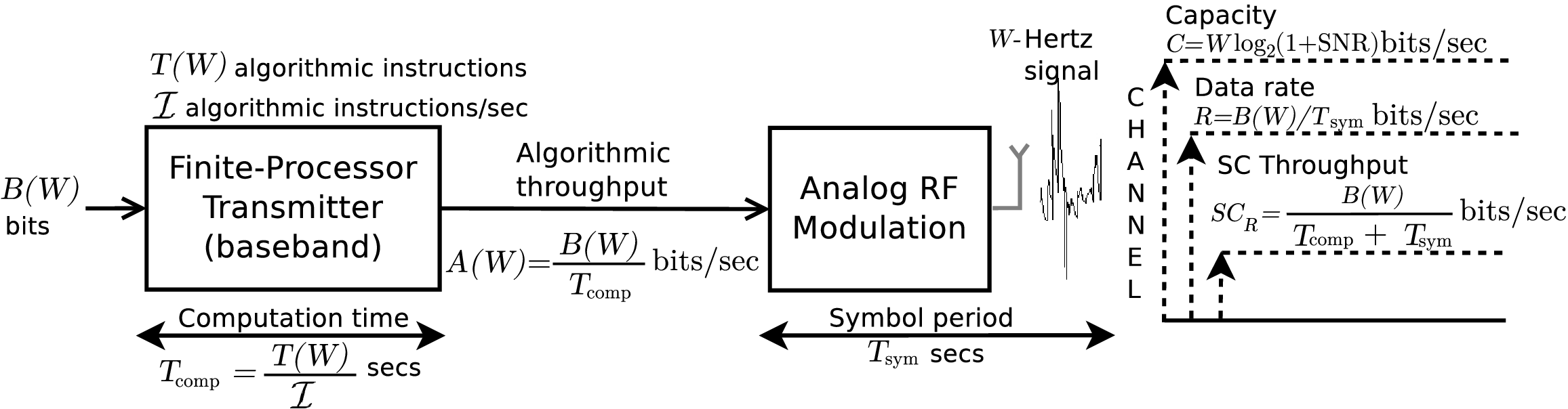}
  \caption{{\color{black} Communication system model of the SC analysis (receiver-side omitted).}
\label{fig:generalsc}
}
\end{figure*}
}

\section{Fundamentals of Spectro-Computational Complexity}\label{subsec:scanalysis}

{\color{black}
Throughout this section, we evolve the classic definitions of information theory
to account for the signal processing time complexity overhead. The resulting analytical 
framework we refer to as the ``SC'' analysis.
The term ``SC'' dates back to our earliest work~\cite{queiroz-cost-ixs-19}, in which we concerned
about the trade-off between spectral efficiency and computational complexity lying in a specific waveform. 
To avoid ambiguity with the classic definitions, we will adopt the same nomenclature of our 
original work to designate each novel enhanced definition. Therefore,
in what follows, we revisit the classic information theory definitions of
data rate $R$ (or throughput, Eq.~\ref{eq:datarate}) and spectral efficiency $SE$ (Eq.~\ref{eq:se}) 
to introduce the enhanced homologous concepts of ``SC throughput'' $SC_{R}$ (Eq.~\ref{eq:sctime}) and
``SC efficiency'' $SC_{SE}$ (Eq.~\ref{eq:sce}), respectively. 
We leave the definition of the ``SC capacity'' to Section~\ref{sec:complimited}.

\subsubsection{System Model and Assumptions}
Fig.~\ref{fig:generalsc} illustrates the transmitter of
our communication system for a $W$-Hertz channel. The analysis
is reciprocal for the receiver side. 
We concern about the throughput experienced by a \emph{particular} 
$B(W)$-bit message.
The message is turned into a baseband signal by a ``finite-baseband transmitter'' that 
can perform $\mathcal{I}>0$ instructions per second. {\color{black} This quantity represents 
all computational resources allocated to executing the algorithmic instructions. It is 
assumed to be finite due to the fundamental limits of chipset manufacturing~\cite{moore-delayed-2003}.
To turn the message into a signal, $T(W)$ computational instructions need to be 
performed (i.e., `time complexity').
The number of bits $B(W)$ sets the input length for the algorithms in the
baseband processor. Since it is a function of $W$ in the
fixed SNR-regime, the complexity $T$ can be written as a function of $W$ as well}.

\subsubsection{Computational Time and Algorithmic Throughput}
The baseband signal processing runtime $T_{\text{comp}}$ and the \emph{algorithmic
throughput}  $A(W)$  of the baseband processor in Fig.~\ref{fig:generalsc} are defined in 
Eq.~(\ref{eq:tcomp}) and Eq.~(\ref{eq:algthroughput}), respectively.
\begin{eqnarray}
T_{\text{comp}}&=&\frac{T(W)}{\mathcal{I}} \quad \text{seconds} \label{eq:tcomp}\\
A(W) &=& \frac{B(W)}{T_{\text{comp}}}= \frac{\mathcal{I}B(W)}{T(W)} \quad \text{bits/second} \label{eq:algthroughput}
\end{eqnarray}

The ``Analog RF Modulation'' block converts the signal to analog and 
performs the carrier modulation. The entire process corresponds to the 
symbol duration $T_{\text{sym}}$ seconds.

\subsubsection{Spectro-Computational Throughput}
Thus, in addition to $T_{\text{comp}}$, the signal carrying the $B(W)$-bit message will 
take $T_{\text{sym}}$ seconds. Based on this, we define the SC data rate (or throughput) $SC_R(W)$ as
\begin{eqnarray}
SC_R(W) =\frac{B(W)}{T_{\text{comp}}+T_{\text{sym}}} =
\frac{B(W)}{\displaystyle\frac{T(W)}{\mathcal{I}}+ T_{\text{sym}}} \quad \text{bits/second} \label{eq:sctime}
\end{eqnarray}
{\color{black}
Note that the algorithmic throughput $A(W)$ (Eq.~\ref{eq:algthroughput})
corresponds to a special case of the SC throughput $SC_R(W)$ (Eq.~\ref{eq:sctime}) when the
channel time $T_{\text{sym}}$ is set to $0$. Besides, as customary 
in the theory of computation, one may also consider the algorithmic performance
independent of a specific hardware. In our case, this stems
by setting $\mathcal{I}=1$}. Thus,
 \begin{eqnarray}
A(W) = SC_R(W) = \frac{B(W)}{T(W)} \quad \text{if $\mathcal{I}=1$ and $T_{\text{sym}}=0$} \label{eq:sc}
\end{eqnarray}
{\color{black}
The particular condition of Eq.~(\ref{eq:sc}) holds, for example,
for an asymptotic analysis of $SC_R(W)$ on $W\to\infty$.
In this case, the constants $T_{\text{sym}}$ and $\mathcal{I}$ can be neglected.}
Thus, the terms \emph{algorithmic throughput} and
\emph{SC throughput} can be interchangeable for the
asymptotic analyses {\color{black} presented throughout this work}.
Similarly, note {\color{black} also} that the rate units 
\text{bits/instruction} and \text{bits/second} are interchangeable
{\color{black} in these cases}.


\subsubsection{Spectro-Computational Efficiency}
Based on the SC throughput (Eq.~\ref{eq:sctime}), we introduce the SC Efficiency 
(SCE) to enhance the classic definition of SE (Eq.~\ref{eq:se}) with
time complexity.  This is given in Eq.~(\ref{eq:sce}).
\begin{eqnarray}
  SC_{SE}(W) &=& \frac{SC_R(W)}{W} \quad \text{bits/second/Hertz}\label{eq:sce}
\end{eqnarray}
{\color{black} Homologously to Eq.~(\ref{eq:sc}), $SC_{SE}(W)$ (Eq.~\ref{eq:sce}) 
can also be expressed in bits/instruction/Hertz if $\mathcal{I}=1$ and $T_{\text{sym}}=0$.
Next, we build on the definitions of this section to formalize a novel
complexity-constrained capacity regime.
}
{\color{black}
\section{Computation-Limited Signals}\label{sec:complimited}
In this section, we build upon the definitions of Section~\ref{subsec:scanalysis} 
to define the SC capacity. Based on this, we formalize a novel capacity regime we
refer to as computation-limited signals. 

\subsection{SC Algorithmic Capacity}\label{subsec:algcapacity}
We define the SC (algorithmic) capacity $SC_{C}(W)$ of a waveform
as the asymptotic upper bound for the SC throughput $SC_R(W)$ of Eq.~(\ref{eq:sctime}), 
i.e., $SC_R(W)=O(SC_C(W))$.
Assuming the fixed SNR regime discussed in Section~\ref{subsubsec:snrregime}, 
the constants $\mathcal{I}$ and $T_{\text{sym}}$ in Eq.~(\ref{eq:sctime}) can be 
neglected as $W\to\infty$. Thus, our analysis follows 
based on $SC_R(W)$ of Eq.~(\ref{eq:sc}). In that case, the upper-bound 
$SC_{C}(W)$ is defined as the ratio 
\begin{eqnarray}
SC_{C}(W) &=& \frac{B_{\text{max}}(W)}{\mathcal{L}(W)}\quad\text{bits/second}\label{eq:scc1}
\end{eqnarray}
and
\begin{eqnarray}
  SC_R(W)&=&O(SC_{C}(W)) 
\end{eqnarray}
In Eq.~(\ref{eq:scc1}), $B_{\text{max}}(W)$ and $\mathcal{L}(W)$ stand for 
the highest and lowest orders of growth that can be assumed for the 
numerator $B(W)$ and the denominator $T(W)$, respectively. {\color{black} 
We assume that the maximum number of bits
passing through the baseband processor grows linearly with the spectrum.
As we discuss in section~\ref{subsec:fairwifi}, this assumption is verified
in practical network standards. Therefore,}  under the fixed SNR regime, it results
\begin{eqnarray}
  B_{\text{max}}(W)=\Theta(C)=cW \quad \text{bits}
\end{eqnarray}

In turn, $\mathcal{L}(W)$ stands for the asymptotic lower bound of a computational 
problem. {\color{black} This lower bound complexity may be hard to derive in some
cases, as is the case of the $N$-point DFT problem required by {\color{black} an
Orthogonal Frequency Division Multiplexing (OFDM) signal
(as discussed in} Section~\ref{sec:examples}).}


\subsection{The Comp-Limited Signal Regime}\label{subsec:complim}
To ensure the SC throughput does not nullify as the channel resources grow, 
the waveform design might satisfy Lemma~\ref{def:conditionscalability}. 

\begin{lemma}[Condition of Scalability]\label{def:conditionscalability}
Under the fixed SNR regime {\color{black}of the Shannon capacity} (discussed in Section~\ref{subsubsec:snrregime}),
the SC throughput (Eq.~\ref{eq:sc}) nullifies as $W\to\infty$ unless $B(W)=\Omega(T(W))$.
\end{lemma}
\begin{proof}
The condition of scalability is such that
\begin{eqnarray}
  SC_R(W\to\infty)=\lim_{W\to\infty} \frac{B(W)}{T(W)} &>& 0, \label{eq:conditionscalability}
\end{eqnarray}
To ensure Ineq.~\ref{eq:conditionscalability} holds, the time complexity must grow as fast as $B(W)$ at most, 
i.e. $B(W)=\Omega(T(W))$. It means either $B(W)=\Theta(T(W))$ (Eq.~\ref{eq:prop1}) 
or $B(W)=\omega(T(W))$ (Eq.~\ref{eq:prop3}). If none of these cases holds, then $B(W)=o(T(W))$ does,
since the conditions of Eq.~\ref{eq:asymptoticnotation} are mutually exclusive. Under this
latter case, it follows from Eq.~(\ref{eq:prop2}) that $SC_R(W\to\infty)=0$. Therefore,
the SC throughput $SC_R(W)$ nullifies as $W\to\infty$ unless $B(W)=\Omega(T(W))$.
\end{proof}

A particular signal implementation may not satisfy Lemma~\ref{def:conditionscalability}.
In some cases, overcoming that is just a matter of devising and implementing asymptotically 
faster signal baseband algorithms e.g., \cite{queiroz-wcl-19}.
{\color{black} We are particularly interested in checking whether the SC capacity
 $SC_C(W)$ remains greater than $0$ as $W$ increases}. If it does not, then it indicates that the signal data rate is constrained by computational resources. This limitation arises because the number of algorithmic instructions cannot be improved beyond the asymptotic lower bound $\mathcal{L}(W)$ present in $SC_C(W)$.
We refer to this as the comp-limited regime (Def.~\ref{def:complimited}).

\begin{definition}[Comp-Limited Signal Regime]\label{def:complimited}
A signal waveform is limited by computation (comp-limited) if its
SC (algorithmic) capacity $SC_C(W)$ (Eq.~\ref{eq:scc1}) {\color{black} nullifies as
$W\to\infty$.}
\end{definition}

Def.~\ref{def:complimited} translates the 
fact that there exist conditions under which the computational
resources of the baseband processor must grow arbitrarily 
(i.e., $\mathcal{I}\to\infty$) to prevent the time complexity
to nullify capacity as $W\to\infty$. \emph{Therefore,  in this regime,
capacity is limited by the computational -- rather than spectrum or power --
resources.}

}

{\color{black}
\section{Examples}\label{sec:examples}
In this section, we demonstrate different use cases of the
SC analysis. {\color{black} In one case study, we show how to perform equitable comparisons 
between signals constrained by different requisites of computational resources.
 In other case study, we study whether OFDM classifies as a comp-limited signal.}

\subsection{Common Parameters}\label{subsec:parameters}
For the analyses of this section, we assume a $N$-subcarrier
 OFDM signal spaced by $\Delta f$~Hz each and a symbol period
of $T_{\text{sym}}=1/\Delta f$ seconds. Given an $M$-point constellation
diagram, {\color{black} the number of bits per subcarrier is $\log_2 M$.
Since $M$ grows on the SNR, this number is constant in the fixed SNR regime. 
Thus, the total number of bits in the OFDM frame solely depends on $N$ and
we will assume it as equal to $N$} without loss of generality.

\subsection{A Fairer Wi-Fi Data Rate Comparison}\label{subsec:fairwifi}
\begin{figure}[t]
\centering
  \includegraphics[scale=0.25]{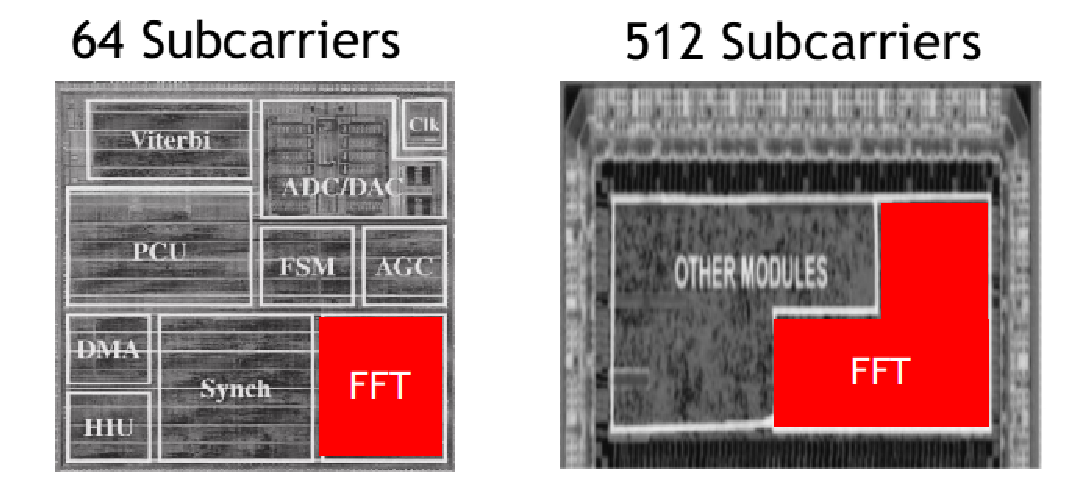}
  \caption{FFT baseband processor comparison: 512-point (right-hand side) vs. 64-point (left-hand side). The 512-point processor {\color{black} produces a $8\times$ faster signal} at the penalty of consuming larger computational resources. We show that this gain nearly halves {\color{black} if the signal data rate also accounts for runtime and both processors are provided the same computational resources.}
\label{fig:fftcomparison}}
\end{figure}

The IEEE 802.11ac standard {\color{black} claims a data rate improvement of}
$8\times$ in comparison to its legacy IEEE 802.11a counterpart.
This results from widening the bandwidth by a factor of $8\times$ 
keeping the legacy OFDM symbol period of $T_{\text{sym}}=3.2$~$\mu$s 
unchanged (without considering cyclic prefix, CP).
However, such an improvement comes at a cost {\color{black} that is not captured by the
classic data rate formula (Eq.~\ref{eq:datarate})}. 
Consider, for example, a DFT
computation of roughly $T_{\text{DFT}}(N)=N\log_2N$ algorithmic instructions~\cite{fastenough-2022}.
Thus, increasing $N$ from $64$ (IEEE 802.11a) to $512$ (IEEE 802.11ac) 
causes a \emph{non-negligible} {\color{black}impact} of roughly $T_{\text{DFT}}(512)/T_{\text{DFT}}(64)=12\times$
in time complexity. {\color{black} 
Despite this increase, the DFT computation time 
$T_{\text{comp}N}=T_{\text{DFT}}(N)/\mathcal{I}_{N}$ (Eq.~\ref{eq:tcomp}) 
($\mathcal{I}_{N}$ denotes the number of instructions per second
of a $N$-point DFT baseband processor) must not exceed 
the symbol duration $T_{\text{sym}}$. Otherwise, the implementation would compromise 
spectral efficiency by introducing idle periods between OFDM symbols~\cite{queiroz-cost-ixs-19}.
Hence, $T_{\text{DFT}}(N)/\mathcal{I}_{N}\leq T_{\text{sym}}=3.2$ must hold. Applying this constraint
to Eq.~(\ref{eq:tcomp}), the \emph{minimum}
performance required by the baseband processors of the $512$-point and $64$-point signals are 
$\mathcal{I}_{512}=T_{\text{DFT}}(512)/3.2=1440$ and $\mathcal{I}_{64}=120$
instructions/microsecond, respectively. This superior demand for computational resources
of the $512$-point signal is clearly illustrated in the die micrograph comparison of 
Fig.~\ref{fig:fftcomparison}\footnote{Illustration art from~\cite{queiroz2022analise}. Due to space constraints, 
we kindly request that the reader refer to the cited work for the technical references about the baseband signal processors.}.

We argue that a fairer data rate comparison should account for
 the signal processing runtime overhead (since it is implicit to the throughput perceived by the upper layers) 
under equitable computational resources, i.e., $\mathcal{I}_{64}=\mathcal{I}_{512}=1440$, in
this case study.
We accomplish this by comparing the SC throughput (Eq.~\ref{eq:sctime}) of 
the $512$-point and $64$-point signals that are $SC_R(512)=512/6.4=80$ bits/microsecond and
$SC_R(64)=64/(0.26+3.2)=18.46$ bits/microsecond, respectively. In this case, the gain $SC_R(512)/SC_R(64)\approx 4.3$
is nearly half of the claimed by $N=512$ setup. \emph{Therefore, 
computational complexity should not be neglected in the data rate analysis
of signals constrained by different computational resources. In this case, the
proposed SC framework can constitute a valuable tool to assist the comparison and
design of novel waveforms}.

\subsection{Is OFDM a Comp-Limited Signal?}
Next, we present a step-by-step analysis of the SC capacity of OFDM to answer whether
it classifies as a comp-limited signal.

{\color{black}
\subsubsection{Asymptotic model and Maximum Number of bits} 
In the fixed SNR regime of Shannon, capacity grows linearly with $W$.
Translated to OFDM, $N\to\infty$ as $W\to\infty$ because $\Delta f$ can be assumed constant
within the OFDM bandwidth $W=N\Delta f$~Hz. Additionally, $M$ remains constant as explained in section~\ref{subsec:parameters}. Consequently, the maximum number of bits in the OFDM frame increases linearly with $N$,
represented as $B_{\text{max-OFDM}}(N\to\infty)=cN$ for some real constant $c>0$.
}
\subsubsection{Complexity of OFDM}
The overall time complexity $T_{\text{OFDM}}(N)$ of the uncoded OFDM
signal results from the sum of its individual procedures, including,
(de)mapping, (I)DFT computation, addition/deletion of CP, signal detection, and 
equalization. Among these, DFT and signal detection are the 
most computationally expensive. However, considering that the input of the data 
signal detection typically involves a small fraction of $N$ (i.e., the pilot
subcarriers), and employing a low-complexity estimator (e.g., the least squares detector) 
with a complexity of $O(N)$, the DFT computation emerges as 
the most complex procedure of OFDM.  Therefore, we equate the asymptotic complexity 
$T_{\text{OFDM}}(N)$ of OFDM to the complexity $T_{\text{DFT}}(N)$ of DFT.
Consequently, as $N$ grows,  the constants and the
complexities of all other procedures can be neglected for the sake of the asymptotic analysis.
It's worth noting that one can proceed with the SC analysis of any OFDM algorithm since the only prerequisites are the waveform parameters (i.e., the number of bits and the symbol duration) and the complexity of the considered algorithm(s).
\subsubsection{SC Capacity of OFDM $SC_{C\text{-OFDM}}(N)$} 
Unfortunately, we are unable to define the SC capacity of OFDM due to the unresolved lower bound complexity of the DFT problem. Nevertheless, with consideration of the conjectures of $\Omega(N\log_2 N)$ and $\Omega(N)$~\cite{fastenough-2022}, Theorem~\ref{th:complimitedofdm} is established.

\begin{theorem}[Comp-Limited OFDM Signal]\label{th:complimitedofdm}
The uncoded OFDM signal is comp-limited unless the $N$-point DFT problem can be solved
in linear time complexity.
\end{theorem}
\begin{proof} Suppose the $\Omega(N\log_2N)$ conjecture holds for the DFT problem,
indicating that the complexity of the FFT cannot be outperformed.
Then, OFDM classifies as a comp-limited signal since its SC capacity is $SC_{C\text{-OFDM}}(N\to\infty)=N/(N\log_2N)=0$.
(Def.~\ref{def:complimited}). By contrast, if $\Omega(N)$ conjecture verifies,
OFDM is not comp-limited since, in this case,  $SC_{C\text{-OFDM}}(N\to\infty)=N/(c_1N)>0$ for
some constant $c_1>0$. \emph{Therefore, the uncoded OFDM signal 
is comp-limited unless the $N$-point DFT problem can be solved
in linear time complexity}. 
\end{proof}
Theorem~\ref{th:complimitedofdm}  further motivates the research at answering whether the $N$-point
DFT problem is $\Omega(N)$, which remains an open question in theory of computation~\cite{fastenough-2022}.

\section{Conclusion and Future Work}\label{sec:conclusion}
In this letter, we proposed a mathematical framework to  
enhance performance indicators from the information theory to 
account for the signal processing time complexity overhead. 
We concerned about why
time complexity is overlooked in classic formulas of information
theory such as data rate.
{\color{black}
One potential explanation we proposed stems from the fact that 
runtime is neglected
 at the expense of other non-temporal performance indicators, such 
as manufacturing cost and chip area. By ignoring the
superior computational resources required by faster signals, classic performance
indicators become susceptible to unfair comparative analyses. 
In a case study, we show that the data rate gain expected from 
widening bandwidth can nearly halve if the narrower signal is provided
the same computational resources required by the wider counterpart.
Thus, our framework can enable fairer comparative analyses 
among signals constrained by different requisites of computational resources.

In other case study, we showed that the signal processing complexity can impose 
a tighter upper bound in comparison the channel capacity. In this case, 
the computational resources can be more crucial for capacity than the channel 
resources. We referred to such novel regime as comp-limited signals. 
Regarding this, we showed that the uncoded OFDM signal
is comp-limited unless the lower-bound complexity of the $N$-point DFT problem verifies as $\Omega(N)$.
Future work may further enhance our model to account for the interplay between time complexity 
and bit error rate in algorithms such as error correction codes.
\bibliographystyle{IEEEtran}
\bibliography{IEEEabrv,refs}

\vfill

\end{document}